\documentclass[vecphys]{svmult}
\usepackage{mathptmx,bm,amssymb}

\newcommand{\upchi}{\raise1pt\hbox{$\chi$}}
\newcommand{\R}{{\mathord{\mathbb R}}}
\newcommand{\C}{{\mathord{\mathbb C}}}
\newcommand{\Z}{{\mathord{\mathbb Z}}}
\newcommand{\N}{{\mathord{\mathbb N}}}

\newcommand{\mfr}[2]{{\textstyle\frac{#1}{#2}}}

\newcommand{\iint}{\mathop{\displaystyle\int\!\!\int}}

\newcommand{\at}{{\mathord{\widetilde{a}}}}

\newcommand{\cD}{{\mathord{\cal D}}}
\newcommand{\cE}{{\mathord{\cal E}}}

\newcommand{\cK}{{\mathord{\cal K}}}
\newcommand{\cM}{{\mathord{\cal M}}}

\newcommand{\cH}{{\mathord{\cal H}}}
\newcommand{\cF}{{\mathord{\cal F}}}
\newcommand{\cR}{{\mathord{\cal R}}}

\newcommand{\Tr}{{\mathord{\rm Tr}}}

\newcommand{\bA}{{\mathord{\vec{\bm A}}}}
\newcommand{\bB}{{\mathord{\vec{\bm B}}}}
\newcommand{\bE}{{\mathord{\vec{\bm E}}}}
\newcommand{\bsigma}{{\mathord{\bm\sigma}}}
\newcommand{\balpha}{{\mathord{\bm\alpha}}}

\begin{document}
\title*{Quantum Coulomb gases} %
\author{Jan Philip Solovej}
\institute{Jan Philip Solovej \at
    Department of Mathematics, University of Copenhagen\\
    Universitetsparken 5,
    DK-2100 Copenhagen, Denmark,
    \email{solovej@math.ku.dk}}
\maketitle

\setcounter{minitocdepth}{3}
\dominitoc

\section{Introduction}

  Ordinary matter is composed of electrons (negatively charged) and
  nuclei (positively charged) interacting via electromagnetic  forces. The
  electric potential between two such particles of charges $Q$ and $Q'$ located at
  $r$ and $r'$ in $\R^3$ is $QQ'/|r-r'|$. The Coulomb potential poses two
  difficulties: (1) The local singularity and (2) its long range. One
  has to understand why the local singularity does not cause
  instabilities and why the long range does not have strong
  macroscopic effects. One of the first major triumphs of the theory
  of quantum mechanics is the explanation it gives of the stability of
  the hydrogen atom (and the complete description of its spectrum) and
  of other microscopic quantum Coulomb systems. 
  It, surprisingly, took
  nearly forty years before the question of stability of everyday
  macroscopic objects was even raised. 
  The rigorous answer to the
  question came shortly thereafter in what came to be known as the
  {\it Theorem on Stability of Matter} proved first by Dyson and Lenard and
  later, in a more simple and transparent way by Lieb and
  Thirring. Since these seminal works, the proof of stability of
  matter has been extended to several different settings, including
  relativistic systems and systems in the presence of dynamic
  electromagnetic fields.

  We will, in particular, discuss the importance of particle
  statistics, i.e., whether the particles are bosons or fermions.
  Using Bogolubov's theory for Bose gases Dyson concluded that charged
  bosonic particles would be macroscopically unstable. That
  Bogolubov's theory gives the correct description of charged Bose
  systems was proved in a series of papers by the author 
  in collaboration with E.H.\ Lieb. 

  Having proved stability of matter the next question is whether
  one can establish the thermodynamics of charged systems, i.e., the
  existence of the thermodynamic limit. This was originally achieved
  by Lieb and Lebowitz. We will describe a new approach to the existence
  of the thermodynamic limit, which applies to many different quantum
  Coulomb systems, possibly in the presence of an underlying lattice
  structure. This generalizes earlier work of Fefferman. 
  
  An important theme in these notes
  is the use of functional inequalities, among which a prominent role
  is played by the Lieb-Thirring inequality.

  Another important
  theme will be how to control the screening of the Coulomb potential.
   
  The notes are organized as follows. In Section~\ref{sec:classical}
  we derive the  classical Hamiltonian for charged particles interacting
  with electromagnetic fields. In Section~\ref{sec:quantization} we
  discuss different ways of quantizing the system. We may quantize
  the particles and leave the fields unquantized or quantize both the
  particles and the fields. Moreover, we may consider the particles
  relativistically or classically. This will give a variety of different
  models. In Section~\ref{sec:stability} we discuss stability 
  both in the sense of stability of individual atoms and in the sense
  of stability of macroscopic matter.
  Finally, in Section~\ref{sec:instability} we discuss situations
  where stability fails, in particular the case of bosonic matter.

\section{Classical ``point'' charges}\label{sec:classical}
We consider $N$ classical particles with charges
$Q_1,\ldots,Q_N\in\R$ situated at points $r_1,\ldots,r_N\in\R^3$.
Dimension $3$ is of course the physical space dimension, but a
natural questions would be whether the discussion of charged system
could be generalized to other dimensions. It is however not entirely
clear what the correct physical questions would be and we will there
restrict the discussion to dimension 3. 

We would ultimately like to consider the particles as point
charges. Unfortunately this will 
however lead to divergencies. In order to avoid this we will initially
assume that each particle is 
given by a charge distribution $Q_i\chi_R(r-r_i)$ where 
$\chi_R(r)=R^{-3}\chi(r/R)$ and $\chi\in C_0(\R^3)$ (a continuous
compactly supported function) with
$\int_{\R^3}\chi=1$.

We will start the discussion of this system of charged particles from
the Lagrangian of the particles and the electromagnetic field.  It
is\footnote{Strictly speaking, even if we use a relativistic kinetic
  energy, this Lagrangian is not relativistically invariant. The reason
  is that we consider the particles as rigid bodies, which do not
  Lorenz contract as they move. We will here ignore this additional
  complication. The Lagrangian in the form given here is that of the
  Abraham model of charged particles \cite{spohn}.}
(using Gaussian units)
\begin{eqnarray*}
  L_R(r_j,\bA,V)&=&\sum_{j=1}^N\Big(T^*_i(\dot{r}_j)
  -Q_j\dot{r}_j\bA*\chi_R(r_j)
  -Q_jV*\chi_R(r_j)\Big)\\
  &&{}+\frac1{8\pi}\int (|\partial_t\bA+\nabla V|^2-|\nabla\times\bA|^2).
\end{eqnarray*}
where 
\begin{itemize}
 \item $\dot{r}_j$ denotes the velocity of particle $j$. 
\item $T_j^*(v)$ is the Legendre transform of the 
  kinetic energy $T_j(p)$ as a function of momentum $p$.
  We 
  assume that the kinetic energy functions $T_j$ 
  are convex functions on $\R^3$.
  For a non-relativistic particle of mass $m_j$ 
  we have $T_j(p)=\frac1{2m_j} p^2$ and thus 
  $T_j^*(v)=\frac12m_j v^2$ and 
  for a relativistic particle it is $T_j(p)=\sqrt{p^2+m_j^2}-m_j$ 
  and hence $T_j^*(v)=m_j-m_j\sqrt{1-v^2}$. We are 
  using units in which the speed of
  light is 1 (a convention we will use throughout these notes).
\item $\bA$ is the vector potential and the magnetic field is
  $\bB=\nabla\times\bA$
\item $V$ is the electric potential. The electric field is 
  $$
  \bE=-\partial_t\bA-\nabla V.
  $$
  \end{itemize}
In order to go to a Hamiltonian description we will choose Coulomb
gauge or more precisely require that
$$
\nabla\cdot\partial_t\bA=0,
$$
i.e., we assume that $\nabla\cdot\bA$ is independent of time. 
We will also for simplicity assume that $\bA$ decays sufficiently fast
that we are allowed to ignore boundary terms when integrating by parts.

With the Coulomb gauge condition we see that 
$$
\bE_\perp=-\partial_t\bA
$$
is the divergence free part of the electric field. 
The total electric field is $
\bE=\bE_\perp-\nabla V
$ and we have 
$$
\int |\bE|^2=\int |\bE_\perp|^2+\int |\nabla V|^2.
$$

The electric potential $V$ is not a dynamic variable in the sense that
that $\partial_t V=\dot{V}$ does not occur in $L_R$. The equation for $V$ is 
$$
0=\frac{\delta L_R}{\delta V}=
-\frac1{4\pi}\Delta V-\sum_{j=1}^NQ_j\chi_R(r-r_j),
$$
where we have used the Coulomb gauge condition. 
We recognize the equation for $V$ as Gauss' law
$$
4\pi\sum_{j=1}^NQ_j\chi_R(r-r_j)=-\Delta V =-
\nabla\cdot(\partial_t\bA+\nabla V)=\nabla\cdot \bE.
$$
In the Hamiltonian formalism this is a constraint equation. 
The solution is 
$$
V(r)=\sum_{j=1}^NQ_j\int\chi_R(r)*|r-r_j|^{-1}dr.
$$
The canonical variable dual to $r_j$ is 
$$
p_j=\nabla_{\dot{r}_j}L_R=\nabla T_j^*(\dot{r}_j)-Q_j\bA*\chi_R(r_j).
$$
The canonical variable dual to $\bA$ is 
$$
\frac{\delta L_R}{\delta \partial_t\bA}=-\frac1{4\pi}\partial_t\bA
=-(4\pi)^{-1}\bE_\perp,
$$
due to the Coulomb gauge condition. 
We then find the Hamiltonian function (where we
assume that $T_j$ is convex such that the double Legendre transform of
$T_j^{**}=T_j$)
\begin{eqnarray*}
  H_R(r_j,p_j,\bA,\bE_\perp)&=&\sum_{j=1}^Np_j\dot{r}_j
  -\frac1{4\pi}\int\partial_t\bA\cdot\bE_\perp
  -L_R(r_j,\bA,V)\\
  &=&\sum_{j=1}^N
  T_j\Big(p_j+Q_j\bA*\chi_R(r_j)\Big)+\sum_{j=1}^NQ_jV*\chi_R(r_j)\\
  &&
  +\frac1{8\pi}\int (|\bE_\perp|^2+|\nabla\times\bA|^2)
  -\frac1{8\pi}\int |\nabla V|^2\\
  &=&\sum_{j=1}^N
  T_j\Big(p_j+Q_j\bA*\chi_R(r_j)\Big)+\frac12\sum_{j=1}^NQ_jV*\chi_R(r_j)\\
  &&
  +\frac1{8\pi}\int (|\bE_\perp|^2+|\nabla\times\bA|^2)\\
  &=&\sum_{j=1}^N
  T_j\Big(p_j+Q_j\bA*\chi_R(r_j)\Big)\\
  &&+\frac12\sum_{j=1}^N\sum_{i=1}^NQ_jQ_i
  \iint\frac{\chi_R(r-r_i)\chi_R(r'-r_j)}{|r-r'|}drdr'\\&&
  +\frac1{8\pi}\int (|\bE_\perp|^2+|\bB|^2).
\end{eqnarray*}
If we subtract the (divergent) self-energy 
$$\sum_{j=1}^N\frac{Q_j^2}{2}
\iint\frac{\chi_R(r)\chi_R(r')}{|r-r'|}drdr'=
R^{-1}\sum_{j=1}^N\frac{Q_j^2}{2}\iint\frac{\chi(r)\chi(r')}{|r-r'|}drdr'
$$ 
and assume that $\bA$ is continuous at $r_j$, we see that  as $R$ tends to
zero $H_R$ converges
pointwise to the Hamiltonian 
\begin{eqnarray}
  H(r_j,p_j,\bA,\bE_\perp)&=&\sum_{j=1}^NT_j\big(p_j+Q_j\bA(r_j)\big)+
  \sum_{1\leq i<j\leq N}\frac{Q_iQ_j}{|r_i-r_j|}\nonumber\\&&
  +\frac1{8\pi}\int (|\bE_\perp|^2+|\bB|^2).\label{eq:Hamcl}
\end{eqnarray}
Here $r_j$ and $p_j$ are dual canonical variables as are 
the field variables $\bA$ and $-(4\pi)^{-1}\bE_\perp$.
Unfortunately, this Hamiltonian is only formal and suffers from
singularities (the field $\bA$ that solves Hamilton's equations will
be singular at $r_j$) leading to severe difficulties in describing the
motion of classical {\it point} charges. 

If external fields $V_{\rm ex}$ and $\bA_{\rm ex}$ are present,
the energy is 
\begin{eqnarray*}
\sum_{j=1}^N \Big(T_j\big(p_j+Q_j(\bA+\bA_{\rm ex})(r_j)\big)+Q_jV_{\rm
  ex}(r_j)\Big)+
  \sum_{1\leq i<j\leq N}\frac{Q_iQ_j}{|r_i-r_j|}\hspace{2cm}\\
  \hfil+\frac1{8\pi}\int (|\bE_\perp|^2+|\bB|^2).
\end{eqnarray*}
\section{Charged quantum gases}\label{sec:quantization}
We now discuss the quantization of the Hamiltonian of charged point
particles. We emphasize that it has not been possible to define a
fully relativistically invariant and causal theory of quantum
electrodynamics (QED) and all the models we describe here are at best
approximations to such a theory (if it exists). All models discussed
here are mathematically well defined (except when otherwise stated explicitly).

There are several levels of quantization that may be
considered. 
\begin{itemize}
\item We can leave the fields $\bA$ and
  $\bE_\perp$ classical and quantize the particles, i.e., describe them
  by a square integrable wave function $\psi(r_1,\ldots,r_N)$. 
\item We may quantize the particles and the fields, i.e., describe the
  particles in terms of a wave function $\psi$ and turn the fields
  into operator valued functions $\bA$ and $\bE_\perp$. This would
  require introducing some cut-off regularization in the fields. 
\item We may second quantize the particles, i.e., let also $\psi$ be
  an operator valued function. This procedure is necessary if we
  consider relativistic particles described by the Dirac operator. 
\end{itemize}

\subsection{Quantized particles and classical fields}
The variables of the system are the $3$-dimensional vector
fields $\bA,\bE_\perp$ (assumed to satisfy appropriate regularity and
decay properties, at least implying that $\bE_\perp$ and
$\bB=\nabla\times\bA$ are square integrable) and the wave function
$\psi$, which is a normalized function in $\bigotimes_{j=1}^N
[L^2(\Omega)]^{\nu_j}$, where $\Omega\subset\R^3$ (say an open set)
and $\nu_j$ is a positive integer counting the
number of internal degrees of freedom of particle $j$, (e.g.\ a particle of
spin $s$ would correspond to $\nu_j=2s+1$). We shall write 
$$
\psi=\psi(r_1,s_1,\ldots, r_N,s_N),\quad r_j\in\Omega,\ s_j=1,\ldots,\nu_j.
$$

The energy of the system is given by 
\begin{equation}\label{eq:NQenergy}
  \cE_N(\psi,\bA,\bE)=\langle\psi,H_N(\bA,\bE_\perp)\psi\rangle,
\end{equation}
where $\langle\psi,\phi\rangle$ refers to the inner product of 
$\psi,\phi\in \bigotimes_{j=1}^N [L^2(\Omega)]^{\nu_j}$ 
and $H_N$ is the (unbounded) operator
(depending on $\bA$ and $\bE_\perp$)
\begin{eqnarray}
  H_N(\bA,\bE_\perp)&=&\sum_{j=1}^N\Big(
  T_j\big(-i\nabla_j+Q_j(\bA+\bA_{\rm ex})(r_j)\big)+
  Q_jV_{\rm ex}(r_j)\Big)+\sum_{1\leq i<j\leq N}
  \frac{Q_iQ_j}{|r_i-r_j|}\nonumber\\&&{}+\frac1{8\pi}\int
  (|\bE_\perp|^2+|\bB|^2).\label{eq:Ham}
\end{eqnarray}
We will throughout be using units in which the reduced Planck constant
$\hbar=1$. The last integral above acts as a ($\bA$ and $\bE_\perp$
dependent) scalar in the Hilbert space. The Hamiltonian
$H(\bA,\bE_\perp)$ depends
also on the exterior fields $V_{\rm ex}$ and $\bA_{\rm ex}$, but we
suppress this in the notation as
these fields usually remain fixed. In fact, we will mostly, and unless
otherwise explicitly stated, assume that the
exterior fields vanish, i.e., $V_{\rm ex}=0$ and $\bA_{\rm ex}=0$.
 
The expectation value $\cE_N(\psi,\bA,\bE_\perp)$ is not defined for all
$\psi$ in the Hilbert space 
$$
\bigotimes_{j=1}^N [L^2(\Omega)]^{\nu_j}.
$$
We will here avoid the discussion of domains of self-adjointness of
the operator $H(\bA,\bE_\perp)$. We will instead restrict attention to
$\psi$ in the the subspace of smooth functions with compact support,
i.e, $[C_0^\infty(\Omega^N)]^{\nu_1\cdots\nu_N}
\subseteq\bigotimes_{j=1}^N [L^2(\Omega)]^{\nu_j}$. 

One of the main issues we will discuss
in these notes is the question of stability, i.e., whether
$\cE_N(\psi,\bA,\bE_\perp)$ is bounded from below independently of $\psi$
(normalized), $\bA$ , and $\bE_\perp$. If such a lower bound holds the
operator $H(\bA,\bE_\perp)$ has a self-adjoint Friedrichs extension
and we are actually making claims about this extension. {F}rom our
point of view the only complication due to considering the restriction
to $C_0^\infty$ is that a possible ground state (a state achieving the
lowest possible energy) is most likely not represented by an element
in  $C_0^\infty$, but only by an element
in the Friedrichs extended domain. We shall, however, not be concerned
with the actual ground states, but only the energy, so we ignore this
issue. 

Since the three coordinates of $-i\nabla_j +Q_j\bA(r_j)$ correspond to
in general non-commuting operators, we must discuss the meaning of
$T_j\big(-i\nabla_j+Q_j\bA(r_j)\big)$. We will, in fact, only consider
examples where the functions $T_j(p)$ can be written in terms of
(possibly matrix-valued) polynomial expressions of $p$ in such a way
that the meaning of $T_j$ (at least on a suitable domain) will be
clear. The examples we will consider are
\begin{itemize}
\item Non-relativistic kinetic energy operators, where 
$T_j(p)=(2m_j)^{-1} p^2$, i.e., the operator is 
\begin{equation}\label{eq:schkin}
T_j\big(-i\nabla_j+Q_j\bA(r_j)\big)=(2m_j)^{-1} (-i\nabla_j+Q_j\bA(r_j))^2.
\end{equation}
We will refer to particles with this kinetic energy as non-relativistic
particles. This is the kinetic energy used when treating 
non-relativistic atoms and molecules or ordinary matter.

\item Relativistic kinetic energy operators, where
$T_j(p)=(p^2+m_j^2)^{1/2}-m_j$, i.e., the operator is 
\begin{equation}\label{eq:pseudokin}
T_j\big(-i\nabla_j+Q_j\bA(r_j)\big)=( (-i\nabla_j+Q_j\bA(r_j))^2+m_j^2)^{1/2}-m_j.
\end{equation}
The square root of an operator is here defined in the spectral
theoretic sense\footnote{The operator inside the square root is
  defined as a self-adjoint operator by Friedrichs extending it from the
  domain of smooth functions with compact support.}.
We will refer to particles with this kinetic energy as relativistic
(or sometimes pseudo-relativistic)
particles.
Both relativistic and non-relativistic particles may have internal
degrees of freedom corresponding to $\nu_j$, being greater than one. 
\item The non-relativistic and relativistic Pauli-operators. These are
  operators acting on two-component vector valued functions given by
  inserting the operator 
  $$\bsigma_j\cdot(-i\nabla_j+Q_j\bA(r_j))$$ into
  the kinetic energy functions above. Here
  $\bsigma=(\sigma^1,\sigma^2,\sigma^3)$ is the vector of $2\times 2$
  Pauli matrices 
  $$
  \sigma^1=\left(\begin{array}{cc}0&1\\1&0
    \end{array}
  \right),\quad
  \sigma^2=\left(\begin{array}{cc}0&-i\\i&0
    \end{array}
  \right),\quad
  \sigma^3=\left(\begin{array}{cc}1&0\\0&-1
    \end{array}
  \right).
  $$
 (The subscript $j$ on $\bsigma$ above
  indicates that it acts on the internal degrees of freedom of
  particle $j$.)
Thus in this case $\nu_j=2$. The resulting kinetic energy operators are
  \begin{equation}\label{eq:paulikin}
    T_j\big(-i\nabla_j+Q_j\bA(r_j)\big)=
    (2m_j)^{-1}(\bsigma_j\cdot(-i\nabla_j+Q_j\bA(r_j)))^2
  \end{equation}
  for non-relativistic Pauli particles and 
  \begin{equation}\label{eq:paulipseudokin}
    T_j\big(-i\nabla_j+Q_j\bA(r_j)\big)=
    \left(\big(\bsigma_j\cdot\left(-i\nabla_j+Q_j\bA(r_j)\right)\big)^2+m_j^2\right)^{1/2}-m_j
  \end{equation}
  for relativistic Pauli particles.
  For the Pauli operator we have the Lichnerowicz' formula
  \begin{equation}\label{eq:lichn}
    \big(\bsigma_j\cdot(-i\nabla_j+Q_j\bA(r_j)\big)^2
    =(-i\nabla_j+Q_j\bA(r_j))^2+Q_j\bsigma_j\cdot\bB(r_j)
  \end{equation}
  and we see that the Pauli operator includes the coupling of the
  particle spin to the magnetic field. 
  
\item We could also consider the $4\times 4$ Dirac operator
  $$
  T_j(p)=\balpha\cdot p+m_j\beta,
  $$ .i.e., 
  $$
  T_j(-i\nabla_j+Q_j\bA(r_j))=\balpha_j\cdot (-i\nabla_j+Q_j\bA(r_j))+m_j\beta_j
  $$
  where $\balpha$ and $\beta$ are standard $4\times 4$ Dirac matrices,
  e.g.,
  $$
  \balpha=\left(\begin{array}{cc}0&\bsigma\\\bsigma&0
    \end{array}
  \right),\quad
  \beta=\left(\begin{array}{cc}I&0\\0&-I
    \end{array}
  \right)
  $$
  using a $2\times2$-block notation. Thus in this case $\nu_j=4$. 
  In contrast to the other types of operators the Dirac operator,
  however, is not positive, in fact, not bounded below and we will
  therefore not be able to treat it unless we second quantize the
  particle fields. We will discuss this briefly below. A different
  approach to deal with the unboundedness from below of the Dirac
  operator is to restrict to the subspace of $L^2(\R^3)^4$ which
  corresponds to the positive spectral subspace of the Dirac
  operator. This approach is called the no-pair theory, but we will
  not discuss it further here. 
\end{itemize} 

\subsection{Statistics of identical particles}

Until now all particles have been considered as
distinguishable, but if we have identical particles the issue of
particle statistics plays an important role. 

The $N$-particle space for $N$-identical particles moving in
$\Omega\subset\R^3$ and with $\nu$ internal degrees of freedom is
$\cH_N=\bigotimes^N [L^2(\Omega)]^\nu$. On $\cH_N$ we define the
orthogonal projections $P_N^\pm$
$$
(P_N^\pm\psi)(r_1,s_1,\ldots,r_N,s_N)=\frac1{N!}\sum_{\sigma\in S_N}
(\pm1)^\sigma\psi(r_{\sigma^{-1}(1)},s_{\sigma^{-1}(1)},\ldots,r_{\sigma^{-1}(N)},s_{\sigma^{-1}(N)}).
$$
They project onto the symmetric $(+)$ or antisymmetric $ (-)$ subspaces. We
denote these subspaces $\cH_N^\pm=P_N^\pm\cH_N$. We will also use the
notation $\cH_N^-=\bigwedge^N [L^2(\Omega)]^\nu$.

In case of $N$ identical particles, i.e., if the operators $T_j$ and
the charges $Q_j=Q$ are  
the same for all $j=1,\ldots,N$,  the Hamiltonian $H_N$ in
(\ref{eq:Ham}) maps the subspaces $\cH_N^\pm$ to themselves and it
makes sense to restrict to these subspaces. We will write 
\begin{equation}\label{eq:fermiboseenergy}
\cE_N^\pm(\psi,\bA,\bE_\perp)
\end{equation}
to emphasize that we restrict to
$\psi\in\cH^\pm_N$. In the symmetric case $(+)$ we say that we have a
system of $N$ {\it bosons} in the antisymmetric case $(-)$ we say that we 
have a system of $N$ fermions. 

It is of course also possible to have mixtures of several species of identical
particles being fermions or bosons 
or even several species of identical particles together with
a number of distinguishable particles. We leave it to the reader to 
work out the structure of the underlying Hilbert space and the
Hamiltonian in the general case. We will look at specific examples
later. 

\subsection{Grand canonical picture}
It is often useful to consider a situation where the particle 
number is not specified at the outset, but where we would instead ask 
what the optimal particle number is in a given situation, e.g., what
number of particles minimizes the energy. This picture is referred to
as the {\it grand canonical picture}. The optimal particle number may
if necessary
be adjusted by adding a term $\mu$ times the particle number to the
Hamiltonian. Such a parameter $\mu$ is called a {\it chemical
  potential}.
We note that this is not the same as adding a constant to the exterior
electric potential $V_{\rm ext}$ as such a constant will multiply the
total charge of the system. 

In order to treat variable particle number we define the bosonic or
fermionic {\it Fock spaces}
$$
\cF^\pm=\cF^\pm((L^2(\Omega))^\nu)=\bigoplus_{N=0}^\infty\cH_N^\pm,
$$
with the convention that $\cH_0^\pm=\C$. The element $1\in\C=\cH_0^\pm$ is
referred to as the vacuum vector and will be denoted by $|{\bf
  0}\rangle$.
For a normalized vector $\psi\in\cF^\pm$ we may write 
$\psi=\bigoplus_{N=0}^\infty\psi_N$, where $\psi_N\in\cH_N^\pm$ with 
$\sum_{N=0}^\infty\|\psi_N\|^2=1$. We say that such a vector
represents a  {\it grand canonical state} and we define the {\it grand canonical energy}
(with chemical potential $\mu$ included)
\begin{equation}\label{eq:GCenergy1sp}
  \cE^\pm(\mu,\psi,\bA,\bE_\perp)=\sum_{N=0}^\infty\cE_N^\pm(\psi_N,\bA,\bE_\perp)
  +\mu N\|\psi_N\|^2.
\end{equation}
As before this energy is not defined for all $\psi$, but we restrict to
$\psi$ corresponding to finitely many particles, i.e., 
$\psi=\bigoplus_{N=0}^M\psi_N$ for some finite
integer $M$ and where each $\psi_N$ is in $C_0^\infty$. 

Again it is possible to consider several species of identical
particles in the grand canonical picture. We leave it to the reader to
write down the Hilbert space and the energy (see also below).

\subsection{Second quantization and quantization of fields}
\label{sec:2ndquant}
We shall here give a brief introduction to second quantization and
discuss how to quantize particle fields and the electromagnetic
fields.

For $f\in L^2(\Omega)^\nu$ we define the annihilation operator 
$\widetilde{a}(f):\cH_N\to\cH_{N-1}$ for $N=1,\ldots$ by
$$
(\widetilde{a}(f)\psi)(r_1,s_1,\ldots,r_{N-1},s_{N-1})
=\sqrt{N}\sum_{s_N=1}^\nu\int
\overline{f(r_N,s_N)}\psi(r_1,s_1,\ldots,r_{N},s_{N})dr_N.
$$
The adjoint of this operator is $\widetilde{a}^*(f):\cH_{N-1}\to\cH_{N}$
given by
$$
(\widetilde{a}^*(f)\psi)(r_1,s_n,\ldots,r_N,s_N)=
\sqrt{N}f(r_N,s_N)\psi(r_1,s_1,\ldots,r_{N-1},s_{N-1}).
$$
We define the bosonic and fermionic annihilation operators 
${a}_\pm(f):\cH_N^\pm\to\cH_{N-1}^\pm$ as the restriction of
$\widetilde{a}(f)$ to the respective subspaces, i.e., 
${a}_\pm(f)=\widetilde{a}(f)_{|\cH_N^\pm}$. The adjoints are
${a}_\pm^*(f):\cH_{N-1}^\pm\to\cH_{N}^\pm$ given by 
${a}_\pm^*(f)=P_N^\pm\widetilde{a}^*(f)_{|\cH_{N-1}^\pm}$.

We may extend ${a}_\pm(f)$ and ${a}_\pm^*(f)$ to operators on the
subspace of the Fock spaces $\cF^\pm$ corresponding to finite particle
numbers, i.e., span$\cup_{M=0}^\infty \bigoplus_{N=0}^M\cH_N^\pm$.
They cannot be extended as bounded operators on the full Fock spaces. 

The extended operators satisfy the famous commutation $(+)$ and
anti-commu\-ta\-tion $(-)$ relations
$$
[a_\pm(f),a^*_\pm(g)]_\pm=\langle f,g \rangle_{L^2(\Omega)^\nu}I
$$
where $[A,B]_\pm=AB\mp BA$ and $I$ is the identity of Fock space (or
rather its restriction to the subspace corresponding to finite particle
numbers.

If $\{f_j\}$ is an orthonormal basis in $L^2(\Omega)^\nu$ we define
the operator valued distributions
$$
\phi_\pm(r,s)=\sum_{j=1}^\infty f_j(r,s)a_\pm(f_j),\quad
\phi^*_\pm(r,s)=\sum_{j=1}^\infty \overline{f_j(r,s)}a^*_\pm(f_j),\
$$
allowing us to write
$$
a_\pm(f)=\sum_{s=1}^\nu\int \overline{f(r,s)}\phi_\pm(r,s) dr,\quad
a^*_\pm(f)=\sum_{s=1}^\nu\int f(r,s)\phi^*_\pm(r,s) dr.
$$
for $f\in L^2(\Omega)^\nu$. Formally we have 
$$
[\phi_\pm(r,s),\phi^*_\pm(r',s')]_\pm=\delta_{ss'}\delta(r-r').
$$
We refer to $\phi$ and $\phi^*$ as field operators. 

Using these field operators we may write the grand canonical Hamiltonian
$\bigoplus_{N=0}^\infty (H_N+\mu N)$, corresponding to identical fermions or
bosons, formally as
\begin{eqnarray*}
\bigoplus_{N=0}^\infty (H_N(\bA,\bE_\perp)+\mu N)&=&
\sum_{s,s'=1}^\nu\int\phi^*_\pm(r,s)\left[T^{ss'}(-i\nabla_r+Q\bA(r))
  +\mu\delta_{ss'}\right]\phi_\pm(r,s')dr
\\&&+\frac{1}2\sum_{ss'}\iint \phi^*_\pm(r,s)\phi^*_\pm(r',s')
\frac{Q^2}{|r-r'|}\phi_\pm(r',s')\phi_\pm(r,s)drdr'\\&&
+\frac1{8\pi}\int (|\bE_\perp|^2+|\bB|^2),
\end{eqnarray*}
where $T^{ss'}$, $s,s'=1,\ldots,\nu$ 
refer to the matrix components of the kinetic energy
operator. It is left as an exercise to the reader to check that the
ordering of $\phi^*$ and $\phi$ exactly gives the correct Hamiltonian
with no-self interactions.

In this formalism it is easy also to write down the grand canonical
operator corresponding to $K$ different species of either fermions or
bosons. For $j=1,\ldots,K$ let $T_j$, $Q_j$, and $\nu_j$ represent the
kinetic energy function, the charge, and the internal degrees of
freedom of species $j$ which is either a fermion or a boson.  The
relevant Hilbert space is
$\cH=\bigotimes_{j=1}^K\cF_j([L^2(\Omega)]^{\nu_j})$ 
where $\cF_j$ is
the Fock space for species $j$. Denoting the field operators for
species $j$ by $\phi_j$ and $\phi^*_j$ the corresponding grand canonical
Hamiltonian (with chemical potential $\mu$ included) is 
\begin{eqnarray*}
H(\mu,\bA,\bE_\perp)&=&\sum_{j=1}^K\sum_{s,s'=1}^{\nu_j}
\int\phi^*_j(r,s) \left[T^{ss'}_j(-i\nabla_r+Q_j\bA(r))
+\mu\delta_{ss'}\right]\phi_j(r,s')dr
\\&&+\frac12\sum_{i,j=1}^K
\sum_{s=1}^{\nu_i}\sum_{s'=1}^{\nu_j}\iint \phi^*_i(r,s)\phi^*_j(r',s')
\frac{Q_iQ_j}{|r-r'|}\phi_j(r',s')\phi_i(r,s)drdr'\\&&
+\frac1{8\pi}\int (|\bE_\perp|^2+|\bB|^2).
\end{eqnarray*}
The subspace on which $H(\mu,\bA,\bE_\perp)$ is defined is the
space corresponding to finitely many particles and where the
restriction to each particle sector is a smooth function with compact
support. Although it is hopefully clear what this means it is rather
complicated to write it down explicitly. For the convenience of the
reader we will nevertheless do this now.  The subspace of $\cH$
corresponding to $N_j$ particles of species $j$, , $j=1,\ldots,K$ is
$\bigotimes_{j=1}^K
P_j\left(\bigotimes^{N_j}[L^2(\Omega)]^{\nu_j}\right)$, where $P_j$
refers to the relevant projection corresponding to the statistics of
species $j$. We may consider this space a subspace of
$[L^2(\Omega^{N_1+\ldots+N_K})]^{\nu_1^{N_1}\cdots\nu_K^{N_k}}$. The
subspace of smooth functions with compact support is
$[C_0^\infty(\Omega^{N_1+\ldots+N_K})]^{\nu_1^{N_1}\cdots\nu_K^{N_k}}$.
Thus the subspace on which we define $H(\mu,\bA,\bE_\perp)$ is
\begin{equation}
\cD=\hbox{span}\!\!\!\bigcup_{M_1,\ldots,M_j=0}^\infty 
\bigoplus_{N_1=0}^{M_1}\cdots\bigoplus_{N_K=0}^{M_K}\left(\bigotimes_{j=1}^K 
P_j\bigotimes^{N_j}L^2(\Omega)^{\nu_j}\right)\bigcap
[C_0^\infty(\Omega^{N_1+\ldots+N_K})]^{\nu_1^{N_1}\cdots\nu_K^{N_k}}.
\end{equation}
The energy of the system with particles being in a state represented
by 
$\psi\in\cD$ 
is denoted
\begin{equation}\label{eq:NQGCenergy}
\cE(\mu,\psi,\bA,\bE_\perp)=\langle\psi,H(\mu,\bA,\bE_\perp)\psi\rangle.
\end{equation}
It is easy to check that this agrees with the definition
(\ref{eq:GCenergy1sp}) in the case of only one species. 

\subsection{Quantization of the electromagnetic field}
We will briefly discuss how to quantize the electromagnetic field.
We will remain in Coulomb gauge and quantize such that $\nabla\cdot\bA=0$. 
This is most conveniently done in momentum space. 

For $k\in\R^3$ choose $e_1(k),e_2(k)\in R^3$ such that 
$e_1(k),e_2(k),k$ form an orthonormal basis. $e_1,e_2$ cannot be
chosen continuously, but this will not cause problems for what we want to say. 

Let $\phi(r,\lambda)$, $\lambda=1,2$ be a bosonic field operator with
two internal degrees of freedom. They are field operators for the
light quanta, i.e., photons. Define the Fourier transformed
operators (Of course they are also simply bosonic field operators)
$$
\widehat\phi(k,\lambda)=(2\pi)^{-3/2}\int e^{-ikr}\phi(r,\lambda)dr,\quad
\widehat\phi^*(k,\lambda)=(2\pi)^{-3/2}\int e^{ikr}\phi^*(r,\lambda)dr.
$$
We define the quantized magnetic vector potential as the operator
valued distribution
\begin{equation}\label{eq:QbA}
\bA(r)=(2\pi)^{-3/2}\sum_{\lambda=1,2}\int_{\R^3}
\sqrt{\frac{2\pi}{|k|}}e_\lambda(k)(e^{ikr}\widehat\phi(k,\lambda)+
e^{-ikr}\widehat\phi(k,\lambda))dk
\end{equation}
and the transversal electric field as
\begin{equation}\label{eq:QbE}
\bE_\perp(r)=i(2\pi)^{-1/2}\sum_{\lambda=1,2}\int_{\R^3}
\sqrt{\frac{|k|}{2\pi}}e_\lambda(k)(e^{ikr}\widehat\phi(k,\lambda)-
e^{-ikr}\widehat\phi(k,\lambda))dk.
\end{equation}
We then find the commutator between the conjugate variables
$$
\left[\bA_i(r),-\frac{1}{4\pi}\bE_{\perp,j}(r'r)\right]_+={\bf P}_{i,j}(r,r'),
$$
where ${\bf P}(r,r')$ is the $3\times 3$-matrix valued integral kernel
of the projection in $L^2(\R^3)^3$ projecting onto  divergence free
vector fields. 

A straightforward (formal) calculation gives for the field energy
\begin{eqnarray*}
  \frac1{8\pi}\int_{\R^3}|\bE_\perp(r)|^2+|\nabla\times\bA(r)|^2dr
  &=&\frac12\sum_{\lambda}\int_{\R^3}|k|(\widehat\phi^*(k,\lambda)\widehat\phi(k,\lambda)\\&&
  +\widehat\phi(k,\lambda)\widehat\phi^*(k,\lambda))dk.
\end{eqnarray*}
This expression however is infinite and we must normal order it to get a
well-defined operator:
$$
\sum_{\lambda}\int_{\R^3}|k|\widehat\phi^*(k,\lambda)\widehat\phi(k,\lambda)dk.
$$
This is the field energy operator of the electromagnetic field on the
Fock space $\cF^+(L^2(\R^3)^2)$. 

\subsection{Non-relativistic QED}
We may now write down the Hamiltonian of non-relativistic QED, i.e.,
of the quantized electromagnetic field coupled to quantized
non-relativistic particles. The particles will be described by the 
non-relativistic kinetic energies (\ref{eq:schkin}) or
(\ref{eq:paulikin}), but since $A$ is now an operator valued
distribution, these operators will not make sense unless we again
introduce the extended charge distribution of the particles. 
The grand canonical non-relativistic QED Hamiltonian for $K$ species
of identical particles is then (ignoring for simplicity  the chemical potential) 
\begin{eqnarray*}
H&=&\sum_{j=1}^K\sum_{s,s'=1}^{\nu_j}
\int\phi^*_j(r,s) T^{ss'}_j(-i\nabla_r+Q_j\bA*\chi_R(r))\phi_j(r,s')dr
\\&&+\sum_{i,j=1}^K\frac{Q_iQ_j}2\sum_{s=1}^{\nu_i}\sum_{s'=1}^{\nu_j}\iint \phi^*_i(r,s)\phi^*_j(r',s')
\frac1{|r-r'|}\phi_j(r',s')\phi_i(r,s)drdr'\\&&
+\sum_{\lambda}\int_{\R^3}
|k|\widehat\phi^*(k,\lambda)\widehat\phi(k,\lambda)dk.
\end{eqnarray*}
This operator is defined on the Hilbert space
$$\Big(\bigotimes_{j=1}^K\cF_j([L^2(\Omega)]^{\nu_j})\Big)
\bigotimes\cF^+([L^2(\R^3)]^2).
$$
The operators $\phi_j$ are field operators for the particles and
$\phi$ is the field operator for the photons. The energy may be
calculated in a state represented by a $\Psi$ in the subspace of the Hilbert space
consisting of $C_0^\infty$ functions of finitely many particles and
photons (we will not write this explicitly this time). The energy is
denoted
\begin{equation}\label{eq:NRQEDenergy}
\cE_{\rm NRQED}(\Psi)=\langle\Psi,H\Psi\rangle.
\end{equation}
As written now the model depends on the regularization parameter $R$. 
The limit as $R$ tends to 0 is not well understood and will require at
least to renormalize the bare mass and charges of the particles. 

\subsection{Relativistic QED Hamiltonian}
As already emphasized a Hamiltonian (or for that matter
any non-perturbative) formulation of QED is non-existent.
Here we simply write down the formal expression for the 
Hamiltonian for the electron-positron field (with charge $e$)
interacting with the 
electromagnetic field: 
\begin{eqnarray*}
&&H_{\rm QED}=\sum_{a,b=1}^4\int
\phi_{\rm e}^*(r,a)(\balpha\cdot(-i\nabla+e \bA(r))+m\beta)_{a,b}\phi_{\rm e}(r,b)dr
\\&&+\frac{e^2}{8}\sum_{a,b=1}^4
\iint\frac{\left[\phi_{\rm e}(a,r),\phi_{\rm e}^*(a,r)\right]_+
\left[\phi_{\rm e}(b,r'),\phi_{\rm e}^*(b,r')\right]_+}{|r-r'|}drdr'\\&&
+\sum_{\lambda}\int_{\R^3}
|k|\widehat\phi^*(k,\lambda)\widehat\phi(k,\lambda)dk.
\end{eqnarray*}
Here $\phi_e$ refers to the fermionic field operator for the electron-positron
field and $\phi$ is the bosonic field operator for the photon field.
The operator $\bA$ is given by (\ref{eq:QbA}).  
Note that we have not distinguished between electrons and positrons,
but that the operator is written in a charge conjugation invariant
way as the density is written as the commutator 
$\frac12\sum_{a=1}^4\left[\phi_{\rm e}(a,r),\phi_{\rm
    e}^*(a,r)\right]_+$.

The operator $H_{\rm QED}$ is ill-defined unless regularizations are
introduced and even in this case it is very difficult to analyze. 
The no-photon situation was studied in the mean-field approximation in
\cite{HLS}.

\section{Stability}\label{sec:stability}
In the previous section we discussed how to define the energy of
states of charged quantum gases in different models. 
 
We have introduced the fixed particle number (or canonical) energy
$\cE_N(\psi,\bA,\bE_\perp)$ in
(\ref{eq:NQenergy}) (or the bosonic or fermionic analogs in 
(\ref{eq:fermiboseenergy})) or the grand canonical energy
$\cE(\mu,\psi,\bA,\bE_\perp)$ in (\ref{eq:NQGCenergy}). We also defined 
the non-relativistic QED energy $\cE_{\rm NRQED}(\Psi)$ in 
(\ref{eq:NRQEDenergy}). 

We will say that a system is {\it stable of the first kind} or {\it
  canonically stable} if the energy $\cE_N(\psi,\bA,\bE_\perp)$ is
bounded below independently of 
$\bA$, $\bE_\perp$, and normalized $\psi$. In this case 
we will call the infimum of $\cE_N(\psi,\bA,\bE)$
the {\it ground state} energy
regardless of whether an actual minimizer (a ground state) exists or
not. Thus the canonical ground state energy of the system is
\begin{eqnarray*}
E_N(\Omega)=\inf\Big\{\cE_N(\psi,\bA,\bE) &\Big|&
\psi\in\Big(\bigotimes_{j=1}^N L^2(\Omega)^{\nu_j}\Big)\cap
C_0^\infty(\Omega^N)^{\nu_1\ldots\nu_N},\ \|\psi\|=1,\\&&
\bA,\bE_\perp\in C_0^\infty(\R^3;\R^3)\Big\}.
\end{eqnarray*}
Note that we are restricting the particles to be in the set $\Omega$
whereas $\bA$ and $\bE$ are unrestricted vector fields in $\R^3$. 
It is immediate to see that we might take $\bE_\perp=0$ in the infimum,
this will however not be the case for quantized fields below. 

The ground state energy of course depends on the types of particles in
the system. We are suppressing this dependence in 
order not to overburden the notation.
 
The ground state is the state of the system
at absolute zero temperature. It is of course also of interest to
study quantum gases at positive temperature corresponding to
minimizing the {\it free energy} we shall however not do this here. 

We could also have chosen to consider the purely static Coulomb
potential and set $\bA=0$, but as we shall see the inclusion of $\bA$
does not really change the treatment in the non-relativistic (and
non-Pauli) case from the points of view discussed here.

We say that a system satisfies {\it stability of the second kind} or
{\it stability of matter} if $N^{-1} E_N(\Omega)$ is bounded below
independently of $N$ for all (open or in some cases sufficiently
regular) $\Omega\subset\R^3$. This is the version of stability mainly
studied in \cite{LSbook}. 

We will here use a slightly stronger notion which we refer to as 
{\it grand canonical stability}. We define the 
{\it grand canonical ground state energy} as
$$
E(\mu,\Omega)=\inf\{\cE(\mu,\psi,\bA,\bE_\perp)\ \Big|\ 
\psi\in\cD,\ \|\psi\|=1,\ 
\bA,\bE_\perp\in C_0^\infty(\R^3;\R^3)\Big\},
$$ 
where $\cE(\mu,\psi,\bA,\bE_\perp)$ was defined in (\ref{eq:NQGCenergy}).
It of course depends on the species of particles. 

We say that a system is grand canonically stable (with chemical
potential $\mu$) if
$$
\inf_{\Omega\subseteq\R^3}|\Omega|^{-1}E(\mu,\Omega)>-\infty.
$$ 
The infimum here is over all open sets $\Omega$ with bounded volume
$|\Omega|$ (or possibly sufficiently regular sets if necessary, but we
will not consider such cases here).

The original proof of stability of matter is due to Dyson and
Lenard~\cite{dl1,dl2} and later by a simpler method by Lieb and
Thirring~\cite{lt}. We will present a proof of grand canonical
stability in a simple case relying on a combination of the two approaches. 

For grand canonically stable systems it is of interest to consider whether
the {\it thermodynamic limit} 
\begin{equation}\label{eq:thermolimit}
  \lim_{\Omega\to\R^3}|\Omega|^{-1}E(\mu,\Omega)
\end{equation}
exists. The limit $\Omega\to\R^3$ can be given a precise meaning in
different ways. Here we shall simply take the simple situation of 
the family of scaled copies $L\Omega$ of a fixed set $\Omega$ and let
the real parameter $L$ tend to infinity. 
\subsection{Stability of the first kind 
for non-relativistic particles}

We shall here prove the stability of the first kind for
non-relativistic particles, i.e., particles with the kinetic energy (\ref{eq:schkin}).
\begin{theorem}[Non-relativistic stability of the first kind]\label{thm:stabfirstkind}
\hfill\\For all $\psi\in C_0^\infty(\Omega^N)^{\nu_1\cdots\nu_N}$ and all
vector fields
$\bA,\bE\in C_0^\infty(\R^3;\R^3)$ we have 
\begin{eqnarray*}
\Big\langle\psi,\Big(\sum_{j=1}^N\frac1{2m_j}(-i\nabla_j+Q_j\bA(r_j))^2
  +\sum_{1\leq i<j\leq N}
  \frac{Q_iQ_j}{|r_i-r_j|}&+&\frac1{8\pi}\int
  (|\bE_\perp|^2+|\bB|^2)\Big)\psi\Big\rangle\\&\geq& -C\|\psi\|^2,
\end{eqnarray*}
where the constant $C>0$ depends only on the number of particles $N$
and their properties, i.e., on  $\nu_1,\ldots,\nu_N\in\N$, $m_1,\ldots,m_N>0$ and $Q_1,\ldots,Q_N\in\R$.
\end{theorem}

This theorem follows easily from the diamagnetic inequality and the Sobolev inequality (see \cite{LLbook}).
\begin{theorem}[Diamagnetic Sobolev inequality]\label{thm:diasob}\hfill\\
For all $f\in C_0^\infty(\R^3)$ and all $\bA\in
C_0^\infty(\R^3;\R^3)$ there is a constant $C>0$ such that 
$$
\int_{\R^3}|(-i\nabla+\bA)f|^2\geq 
\int_{\R^3}|\nabla|f||^2\geq C\Big(\int_{\R^3}|f|^6\Big)^{1/3}.
$$
\end{theorem}
An immediate corollary of this result (using simply H\"older's inequality) is the
following bound on one-body Schr\"odinger operators.
\begin{corollary}[Lower bound on Schr\"odinger operator]\label{cl:schlow}\hfill\\
For all $f\in C_0^\infty(\R^3)$, $\bA\in
C_0^\infty(\R^3;\R^3)$, $0\leq V_1\in L^{5/2}(\R^3)$, and $0\leq V_2\in L^\infty(\R^3)$ we have 
$$
\langle f,\big((-i\nabla+\bA)^2-V_1-V_2\big)f\rangle_{L^2}\geq -C\Big(\int V_1^{5/2}+\|V_2\|_\infty\Big)\|f\|^2_{L^2}.
$$
\end{corollary}
We leave it to the reader to prove Theorem~\ref{thm:stabfirstkind}
from this corollary and the observation that $|r|^{-1}\in
L^{5/2}(\R^3)+L^\infty(\R^3)$.

The stability of the first kind holds even if the field energy 
$$
\frac1{8\pi}\int|\bE_\perp|^2+|\bB|^2
$$
is ignored. Moreover, it also holds if $\bA$ is quantized, i.e., if we
replace $\bA(r)$ by the operator (\ref{eq:QbA}). This last statement
follows since $\bA(r)$ is a commuting family (indexed by $r$) and thus
may be considered as a classical field. ´

\subsection{Grand canonical stability}
We turn to the question of grand canonical stability. 
We will study this in the simple special case of two species of
identical fermions with opposite charges. For grand canonical
stability it is not necessary that all particles are fermions. It is,
in fact,
enough that all particles with one sign of the charge, i.e., say, all
negatively charged particles form a
collection of finitely many species of fermions. Stability of matter
in this more general setting was proved in \cite{dl1,dl2,lt} (see
also~\cite{l-rev}) and the
case of grand canonical stability was treated in \cite{HLS3}.

One of the main ingredients in the proof of grand canonical stability
is the use of the celebrated Lieb-Thirring inequality \cite{lt} (see
also \cite{LSbook}) which replaces the Sobolev inequality which we
used in the proof of stability of the first kind.

\begin{theorem}[Lieb-Thirring inequality]\label{thm:LT}\hfill\\ Assume $0\leq V\in
  L^{5/2}(\R^3)$ and $\bA\in C_0(\R^3;\R^3)$ then for all $N$ 
we have on the antisymmetric subspace
  $\cH_N^-=\bigwedge^N[L^2(\Omega)]^\nu$ the operator inequality 
$$
\sum_{j=1}^N\left(\frac1{2m}(-i\nabla_j+\bA(r_j))^2-V(r_j)\right)\geq
-Cm^{3/2}\nu\int V^{5/2},
$$
for a universal constant $C>0$.
In particular, it is independent of the number $N$ of particles. 
\end{theorem}
Note the apparent similarity between the Lieb-Thirring inequality and 
the Corollary~\ref{cl:schlow}, to the Sobolev inequality. The
important difference is that Corollary~\ref{cl:schlow} would only 
imply that
$$ 
\sum_{j=1}^N\left(\frac1{2m}(-i\nabla_j+\bA(r_j))^2-V(r_j)\right)\geq
-Cm^{3/2}N\int V^{5/2},
$$
which, in fact, holds on all of $\cH_N$ (left as an exercise for the
reader). The lower bound with a constant independent of $N$ holds only
on the fermionic subspace. 

The Lieb-Thirring inequality relates the energy of a gas of independent
particles to the corresponding classical energy. The classical
energy (ignoring internal degrees of freedom) would indeed be
$$
\iint\limits_{\frac1{2m}(p^2+\bA(r))^2-V(r)\leq0} \frac1{2m}(p^2+\bA(r))^2-V(r)
dr dp=-\frac{8\pi}{15} m^{3/2}\int V^{5/2}.´
$$
As a consequence of the Lieb-Thirring inequality we have the following
lower bound on the kinetic energy of $N$ fermions confined to move in a
bounded volume.
\begin{corollary}\label{cl:LT} If the open set $\Omega$ has finite
  volume $|\Omega|$ then
  in $\bigwedge^N[L^2(\Omega)]^\nu$ we have a universal constant
  $C>0$ such that
\begin{equation}\label{eq:kinenineq}
  \sum_{j=1}^N\frac1{2m}(-i\nabla_j+\bA(r_j))^2
  \geq Cm^{-1}\nu^{-2/3}N^{5/3}|\Omega|^{-2/3}.
\end{equation}
\end{corollary}
\begin{proof} If we use the Lieb-Thirring inequality with a constant
  potential $V$ we obtain
$$
\sum_{j=1}^N\frac1{2m}(-i\nabla_j+\bA(r_j))^2\geq NV-
Cm^{3/2}\nu V^{5/2}|\Omega|
$$
which gives the estimate above after optimization in $V$.
\end{proof}
We now consider the situation with two species of identical
non-relativistic fermions
with masses $m_\pm>0$ and charges $\pm Q_\pm$ where $Q_\pm>0$. For
simplicity we assume that there are no internal degrees of freedom, i.e.,
$\nu_\pm=1$. In this
case the Hamiltonian with particle numbers $N_\pm$ for the two species is 
\begin{eqnarray*}
  H_{N_+,N_-}&=&\sum_{j=1}^{N_+}
  \frac1{2m_+}(-i\nabla_j+ Q_+\bA(r_j))^2
  +\sum_{j=N_++1}^{N_++N_-}\frac1{2m_-}(-i\nabla_j-Q_-\bA(r_j))^2+V_{\rm
    C}\\&&+\frac1{8\pi}\int (|E_\perp|^2+|B|^2)
\end{eqnarray*}
where the Coulomb energy is
$$V_{\rm C}=
   -\sum_{j=1}^{N_+}\sum_{i=N_++1}^{N_++N_-}\frac{Q_+Q_-}{|r_i-r_j|}
   +\sum_{1\leq i<j\leq N_+}\frac{Q_+^2}{|r_i-r_j|}
   +\sum_{N_+< i<j\leq N_++N_-}\frac{Q_-^2}{|r_i-r_j|}.
$$
Note that we have numbered the positively charged
particles  $1,\ldots,N_+$ and the negatively charged particles 
$N_++1,\ldots, N_++N_-$.
The Hamiltonian acts on the subspace 
$$
\cD=\left(\bigwedge^{N_+}L^2(\Omega)\otimes
  \bigwedge^{N_+}L^2(\Omega)\right)\cap C_0^\infty(\Omega^{N_++N_-}).
$$
\begin{theorem}[Simple case of grand canonical stability]\hfill\\ The grand
  canonical energy in the finite volume set $\Omega\subseteq\R^3$ 
\begin{eqnarray*}
  E(\mu,\Omega)=\inf\Big\{\langle
    \psi,H_{N_+,N_-}\psi\rangle+\mu(N_+N_-)&|& \psi\in \cD,\
    \|\psi\|=1,\\&& \bE_\perp,\bA\in C_0^\infty(\R^3;\R^3)\Big\}
\end{eqnarray*}
satisfies the stability bound
$$
E(\mu,\Omega)\geq -C(\mu,m_\pm,Q_\pm)|\Omega|,
$$
for a constant $C(\mu,m_\pm,Q_\pm)>0$ depending only on $\mu,m_\pm,Q_\pm$.
\end{theorem}
\begin{proof}
We define the distance from particle $j$ to the nearest particle of
the opposite charge, i.e., 
\begin{eqnarray*}
  \delta_j&=&\delta_j(r_1,\ldots,r_{N_++N_-})\\&=&
  \left\{\begin{array}{ll}\min_{i=N_++1,\ldots,N_++N_-}|r_i-r_j|,&\hbox{if
      }j=1,\ldots,N_+\\
      \min_{i=1,\ldots,N_+}|r_i-r_j|,&\hbox{if
      }j=N_++1,,\ldots,N_++N_-
    \end{array}\right..
\end{eqnarray*}
Let $\chi_j=\frac6{\pi\delta_j^3}{\bf 1}_{B(r_j,\delta_j/2)}$, where
$B(r_j,\delta_j/2)$ denotes the ball centered at $r_j$ with radius
$\delta_j/2$ and ${\bf 1}_{B(r_j,\delta_j/2)}$ is its characteristic
function. Note that $\int\chi_j=1$. 

We will use the following two observations:

\noindent{\bf Observation 1:}
$$
\sum_{j=1}^{N_+}\sum_{i=N_++1}^{N_++N_-}\frac{Q_+Q_-}{|r_i-r_j|}
=\sum_{j=1}^{N_+}\sum_{i=N_++1}^{N_++N_-}{Q_+Q_-}
\iint\frac{\chi_j(r)\chi_i(r')}{|r-r'|}drdr'
$$

\noindent{\bf Observation 2:}
$$
\sum_{1\leq i<j\leq N_+}\frac{Q_+^2}{|r_i-r_j|}
\geq \sum_{1\leq i<j\leq N_+}{Q_+^2}
\iint\frac{\chi_j(r)\chi_i(r')}{|r-r'|}drdr'
$$
and likewise for the $Q_-^2$-terms .

The observations follow from Newton's Theorem:
$$
\frac6{\pi\delta^3}\int_{|r'|<\delta/2}|r-r'|^{-1}dr'
=\left\{\begin{array}{ll} |r|^{-1},& \hbox{if }|r|>\delta/2\\
    \delta^{-1}(3-4|r|^2\delta^{-2}),& \hbox{if }|r|<\delta/2\\
  \end{array}\right.\leq |r|^{-1}.
$$
{F}rom the two observations above we arrive at 
the following lower bound on the Coulomb energy
$$
V_{\rm C}\geq \frac12\iint\frac{\rho(r)\rho(r')}{|r-r'|}drdr'
-\frac{12}{5}\sum_{j=1}^{N_+}Q_+^2\delta_j^{-1}-\frac{12}{5}\sum_{j=N_++1}^{N_++N_-}Q_-^2\delta_j^{-1},
$$
where we introduced the smeared charge density 
$$
\rho(r)=\sum_{j=1}^{N_+}Q_+\chi_j(r)-\sum_{j=N_++1}^{N_++N_-}Q_-\chi_j(r)
$$
and used that 
$$
\iint\frac{\chi_j(r)\chi_j(r')}{|r-r'|}drdr'=\frac{12}{5}\delta_j^{-1}.
$$
Using now the positive type (i.e., positivity of the Fourier
transform) of the Coulomb kernel we find 
$$
V_{\rm C}\geq
-\frac{12}{5}\sum_{j=1}^{N_+}Q_+^2\delta_j^{-1}
-\frac{12}{5}\sum_{j=N_++1}^{N_++N_-}Q_-^2\delta_j^{-1}.
$$ 
A similar application of the positive type of the Coulomb kernel goes
back to an early paper of Onsager~\cite{on}, who might have been the
first to address the issue of grand canonical stability. 
Better lower bounds on the Coulomb energy can be derived by more
sophisticated use of the same ideas (see e.g.\ \cite{b,ly,LSbook}).

We are led to the following lower bound on the Hamiltonian
$$
H_{N_+,N_-}\geq H_{N_+}+H_{N_-}+\mu (N_++N_-)
$$
where 
$$
H_{N_+}=\sum_{j=1}^{N_+}\frac1{2m_+}(-i\nabla_j+ Q_+\bA(r_j))^2-
\frac{12}{5}\sum_{j=1}^{N_+}Q_+^2\delta_j^{-1}
$$ 
and likewise for $H_{N_-}$.
Observe now that for $j=1,\ldots,N_+$ the length $\delta_j$ depends on the
position $r_j$ and the positions $r_{N_++1},\ldots,r_{N_++N_-}R$ 
of the negatively charged particles
but not on the positions of the other
positively charged particles. In other words we may write
$$
-\frac{12}{5}\sum_{j=1}^{N_+}Q_+^2\delta_j^{-1}
=-\frac{12}{5}\sum_{j=1}^{N_+}Q_+^2\delta(r_j)^{-1},
$$
where $\delta(r)=\min_{i=N_++1,\ldots,N_++N_-}|r_i-r|$. We thus have a
potential parameterized by the positions of the negatively charged
particles. This observation allows us to use the Lieb-Thirring
inequality Theorem~\ref{thm:LT}. If we choose a parameter $R$ (to be
optimized over) and divide the space into the region where
$\delta(r)<R$ (a union of $N_-$ possibly intersecting balls of radius R) 
and $\delta(r)>R$ we obtain from the Lieb-Thirring inequality
\begin{eqnarray*}
  H_{N_+}&\geq&\frac12\sum_{j=1}^{N_+}\frac1{2m_+}(-i\nabla_j+
  Q_+\bA(r_j))^2\\&&
  -CQ_+^5m_+^{3/2}\Big(N_-\int_{|r|<R}|r|^{-5/2}dr+R^{-5/2}|\Omega|\Big)\\
  &\geq&Cm_+^{-1}N_+^{5/3}|\Omega|^{-2/3}-CQ_+^5m_+^{3/2}(N_-R^{1/2}+|\Omega| R^{-5/2})\\
  &=&Cm_+^{-1}N_+^{5/3}|\Omega|^{-2/3}-CQ_+^5m_+^{3/2}N_-^{5/6}|\Omega|^{1/6},
\end{eqnarray*}
where we saved half of the kinetic energy in the first inequality and
estimated it by Corollary~\ref{cl:LT} in the second inequality. 
Finally, we optimized over the parameter $R>0$.
Since the corresponding estimate holds for $H_{N_-}$ we finally get
the lower bound
\begin{eqnarray*}
  H_{N_+,N_-}&\geq&
  Cm_+^{-1}N_+^{5/3}|\Omega|^{-2/3}+Cm_-^{-1}N_-^{5/3}|\Omega|^{-2/3}\\&&
  -CQ_+^5m_+^{3/2}N_-^{5/6}|\Omega|^{1/6} 
  -CQ_-^5m_-^{3/2}N_+^{5/6}|\Omega|^{1/6}+\mu(N_++N_-)\\
  &\geq&-C(\mu,m_\pm,Q_\pm)|\Omega|,
\end{eqnarray*}
where we have minimized in $N_\pm$. 
We leave it to the reader to determine the exact form of the constant
$C(\mu,m_\pm,Q_\pm)$.
\end{proof}
The same proof would work also if periodic external electric and
magnetic fields were present, e.g., a situation describing a
crystal structure.
 
As should also be clear from the proof the field energy 
$$
\frac1{8\pi}\int|\bE_\perp|^2+|\bB|^2
$$
plays no role for stability in the present case.  
Moreover, as in the case discussed for stability of the first kind
we could also have considered $\bA$ quantized.

\subsection{Existence of the thermodynamic limit}
We will briefly discuss existence of the thermodynamic limit
(\ref{eq:thermolimit}). This was first proved by Lieb and Lebowitz 
\cite{lile} for the case of several species of particles where all the
species of, say, negatively charged particles are fermions. 
The method does not allow for an exterior periodic potential or
magnetic field. In particular, the method does work in the case where
the nuclei are confined to a periodic crystal arrangement. This case was
later treated by Fefferman in \cite{Fe}. In \cite{HLS2,HLS3} an
abstract method was developed to conclude existence of thermodynamic
limits for Coulomb systems in great generality including periodic
background potentials.

Indeed, the method relies on establishing general abstract properties of
the energy function that implies existence of the thermodynamic limit.

We will just give a brief overview of the method. For the details and 
more precise definitions and assumptions we refer to
\cite{HLS2,HLS3}. 

Let $\cM= \{\Omega\subset\R^3\
  \hbox{open and bounded}\}$ and consider a map $E:\cM\to\R$ with the
  following properties.  Given a  
  function $\alpha:[0,\infty)$ with
  $\lim_{\ell\to\infty}\alpha(\ell)=0$, a
  subset $\cR\subseteq\cM$ of sufficiently regular sets, constants
  $\kappa,\delta>0$, and a
  reference set $\triangle\in\cR$, such that

\textbf{(A1)} {\it (Normalization).} $E(\emptyset)=0$.

\medskip

\textbf{(A2)} {\it (Stability).} $\forall \Omega\in{\cM}$, 
${} E(\Omega)\geq -\kappa|\Omega|$.

\medskip

\textbf{(A3)} {\it (Translation Invariance).} $\forall
\Omega\in{\cR}$, $\forall z\in\Z^3$, ${}
E(\Omega+z)=E(\Omega)$.

\medskip

\textbf{(A4)} {\it (Continuity).} $\forall\Omega,\Omega'\in{\cR}$,
with $\Omega'\subseteq\Omega$
and $\hbox{d}(\partial\Omega,\partial\Omega')>\delta$,

\centerline{${} E(\Omega)\leq E(\Omega')+\kappa|\Omega\setminus\Omega'|+|\Omega|\alpha(|\Omega|).$}

\medskip

\textbf{(A5)} {\it (Subaverage Property).} For all $\Omega\in{\cM}$, we have
\begin{equation}{} 
  E(\Omega)\geq \frac{1}{|\ell\triangle|}\int_{\R^3\rtimes
    SO(3)}E\big(\Omega\cap g\cdot(\ell\triangle)\big)\,d\lambda(g)-|\Omega|_{\rm r}\,\alpha(\ell)
\label{sliding_equation}
\end{equation}
where $d\lambda$ is the Haar-measure of $\R^3\rtimes
  SO(3)$, (the group of isometries of $\R^3$) and
$|\Omega|_{\rm r}:=\inf\{|\tilde\Omega|,\quad
\Omega\subseteq\tilde\Omega,\quad \tilde\Omega\in{\cR}\}$ is a 
regularized volume of $\Omega$.

If $E:\cM\to\R$ satisfies (A1--A5) then it is not very difficult to
show that the thermodynamic limit
$$
\lim_{\ell\to\infty} |\ell\triangle|^{-1}E(g\ell\triangle)
$$
exists for all $g\in\R^3\rtimes SO(3)$, i.e., it exists
for all rotations or translations of the reference set $\triangle$.
Under slightly more restrictive assumptions which we will not repeat
here the limit holds for a very large class of regular sets. 

We see that (A2) is grand canonical stability. The difficult property
to establish for Coulomb systems is (A5). For $\triangle$ being a
simplex it is a consequence of the
following result of Graf and Schenker \cite{GS} generalizing a
somewhat similar
estimate by Conlon, Lieb and Yau \cite{CLY}:

\begin{theorem}[Graf-Schenker inequality]\label{Prop_Graf_Schenker}\hfill\\ Let
  $\triangle$ be a simplex in $\R^3$. There exists a constant $C$ such
  that for any $N\in\N$, $Q_1,...,Q_N\in\R$, $r_1,\dots,r_N\in\R^3$ and any
  $\ell>0$,
  \begin{eqnarray*}
    \sum_{1\leq i<j\leq N} \frac{Q_iQ_j}{|r_i-r_j|}&\geq &
    \frac{1}{|\ell\triangle|}\int_{\R^3\rtimes
      SO(3)}
    \sum_{1\leq i<j\leq N}
    \frac{Q_iQ_j{\bf 1}_{g\ell\triangle}(r_i)
      {\bf 1}_{g\ell\triangle}(r_j)}{|r_i-r_j|} d\lambda(g)
    \\&&-\frac{C}{\ell}\sum_{j=1}^NQ_j^2.
  \end{eqnarray*}
\end{theorem}
This inequality follows by proving that the function 
$$
F(r,r')=\int_{\R^3\rtimes
  SO(3)}
    {\bf 1}_{g\ell\triangle}(r_i)
      {\bf 1}_{g\ell\triangle}(r_j)d\lambda(g),
$$
is of the form $F(r,r')= g(|r-r'|)$ where $g$ is such that 
$|r|^{-1}(1-g(|r|)$ has positive Fourier transform. Recall that 
for a function $f$ of positive type 
$$\sum_{1\leq i<j\leq N}Q_iQ_jf(r_i-r_j)
\geq -\sum_{j=1}^NQ_j^2f(0).$$ 

\section{Instability}\label{sec:instability}
\subsection{Examples of instability of the first kind}

As an example of a system that can show instability of the first kind
we consider two relativistic particles with masses $m_1=m_2=1$ and
charges $Q_1=-1$ and
$Q_2=Q>0$. 
The kinetic energy is given by (\ref{eq:pseudokin}) and we simply set
$\bA=0$. Thus the Hamiltonian is 
$$
H=\sqrt{-\Delta_1+1}-1+\sqrt{-\Delta_2+1}-1-\frac{Q}{|r_1-r_2|}
$$
acting on the smooth compactly supported functions in 
$L^2(\R^3)\otimes L^2(\R^3)$. Let $\psi\in C^\infty_0(\R^6)$ be
normalized, i.e., its square integral is one. 
and define $\psi_\ell(r_1,r_2)=\ell^{-3}\psi(r_1/\ell,r_2/\ell)$
for $\ell>0$. Note that $\psi_\ell$ is still normalized for all $\ell$.
Then 
$$
\langle \psi_\ell,H\psi_\ell\rangle =\ell^{-1}\left\langle
\psi,
\left(\sqrt{-\Delta_1+\ell^{2}}-\ell+\sqrt{-\Delta_2+\ell^{2}}-\ell
-\frac{Q}{|r_1-r_2|}
\right)\psi\right\rangle.
$$
Thus if we let $\ell$ tend to zero
$$
\lim_{\ell\to0}\ell\langle \psi_\ell,H\psi_\ell\rangle
= \left\langle
\psi,
\left(\sqrt{-\Delta_1}+\sqrt{-\Delta_2}
-\frac{Q}{|r_1-r_2|}
\right)\psi\right\rangle.
$$ 
If $Q$ is large enough we find that the right side is negative and
hence for such a $Q$
$$
\lim_{\ell\to0}\langle \psi_\ell,H\psi_\ell\rangle=-\infty
$$
and the system is not stable of the first kind. 

On the other hand, if the negatively charged particles belong to a
finite number of fermionic species and if the number of fermionic species,
the maximal negative charge and the maximal positive charge satisfy
appropriate bounds, then stability of matter holds
\cite{C,fdl,LSbook,lls-tf,ly}.

A similar situation happens for non-relativistic particles interacting
with magnetic fields according to the Pauli operator
(\ref{eq:paulikin}).
Consider as an example a Hamiltonian for two particles of mass $m=1$
and charges $Q_1=Q>0$ and $Q_2=-Q<0$:
\begin{eqnarray*}
  H(\bA)&=&\frac12(\bsigma\cdot(-i\nabla_1-Q\bA(r_1)))^2+
  \frac12(\bsigma\cdot(-i\nabla_2+Q\bA(r_2)))^2-\frac{Q^2}{|r_1-r_2|}
  \\&&+\frac1{8\pi}\int|\nabla\otimes\bA|^2,
\end{eqnarray*}
where we have chosen $\bE=0$ (which is the energetically best
choice). 
The instability in this case relies on the existence 
(see \cite{ES,loya}) of 
a non-zero $\widetilde\psi\in L^2(\R^3)$ and 
a magnetic field $\widetilde\bA$ with
$\int|\nabla\otimes\widetilde\bA|^2<\infty$
such that 
$$
\frac12(\bsigma\cdot(-i\nabla_1-\widetilde\bA(r_1)))^2\widetilde\psi=0.
$$
We may assume that $\widetilde \psi$ is normalized.
If for $\ell>0$ we set 
$$\psi_\ell(r_1,r_2)=
\ell^{-3}\widetilde\psi(r_1/\ell)\widetilde\psi(r_2/\ell)$$ 
(which is
also normalized)
and $\bA_\ell(r)=(Q\ell)^{-1}\widetilde\bA(r/\ell)$ we obtain for the
energy expectation
$$
\ell\langle \psi_\ell, H(\bA_\ell)\psi_\ell)=-
\left\langle\psi_{\ell=1},\frac{Q^2}{|r_1-r_2|}\psi_{\ell=1}\right\rangle
    +\frac1{8\pi Q^2}\int|\nabla\otimes\widetilde\bA|^2 
  .
$$
Again we see that if $Q$ is large enough the right side is
negative and hence for such a $Q$ we  have
as before  $\lim_{\ell\to\infty}\langle \psi_\ell,
H(\bA_\ell)\psi_\ell)=-\infty$.
As for the relativistic case stability of matter also holds in this
case under appropriate conditions \cite{fef,lls}.
This problem with a quantized field has been treated in
\cite{bfg,ffg}, the relativistic case with classical fields
is considered in \cite{lss},  and the relativistic case with quantized
field in \cite{lilo2}.

\subsection{Fermionic instability of the second kind}
As the final topic of these notes we will discuss
instability of the second kind. 

We will first make a very simple general remark about instability
of many-body systems with attractive interactions 
which has nothing to do with charged systems and holds even for
fermions.  
\begin{theorem}[Fermionic instability for attractive 2-body potentials]\hfill\\
 Assume that the potential $W:\R^n\to\R$ satisfies $W(r)\leq -c<0$ 
  for all $r$ in a ball around the origin. Consider the $N$-body operator
$$
H_N=\sum_{j=1}^N-\frac12\Delta_j+\sum_{1\leq i<j\leq N}W(r_i-r_j)
$$
acting in the fermionic Hilbert space $\bigwedge^N L^2(\R^n)$.
If $n\geq3$ then $H_N$
cannot be stable of the second kind, i.e., we can find a a sequence of
normalized vectors $\psi_N\in\bigwedge^N L^2(\R^n)$ such that
$$
\lim_{N\to\infty}N^{-1}\langle\psi_N,H_N\psi_N\rangle=-\infty.
$$
\end{theorem}
\begin{proof}
Assume that $W(r)\leq -c<0$ on the ball of radius $R$ centered at the
origin. Define $\psi_N$ as the (normalized) Slater determinant
$$
\psi_N(r_1,\ldots,r_N)=(N!)^{-1/2}\det(u_i(r_j))_{i,j=1}^N
$$
where $u_j$, $j=1,\ldots,N$ are  orthonormalized eigenfunctions 
corresponding to the $N$ lowest eigenvalues of the negative
Laplacian with Dirichlet boundary conditions for the largest cube centered at
the origin and contained in the ball of radius $R$. 
We extend the functions to be 0 outside the cube.
The functions $u_j$ are explicit and can be written in terms of  
sines and cosines.  It is a simple straightforward calculation to show
that in all dimensions $n$ there is a constant $C_n$ such that 
$$
\langle\psi_N,\sum_{j=1}^N(-\frac12\Delta_j)\psi_N\rangle
\leq C_n N^{(n+2)/n}R^{-2}.
$$
Comparing with Corollary~\ref{cl:LT} (written for the case $n=3$)
we see that there is always a similar lower bound. 

Thus 
$$
 N^{-1}\langle\psi_N,H_N\psi_N\rangle\leq C_n N^{2/n}R^{-2}
-\frac12 (N-1)c.
$$
We see that instability occurs when $n>2$.
\end{proof}

\subsection{Instability of bosonic matter}
For matter consisting of charged particles 
we have discussed that the fermionic property ensures 
grand canonical stability. In this final section
we will show that the fermionic property is indeed 
a necessity as stability fails for bosons. 

We consider two species of bosons with masses
$m_\pm=1$, $Q_+=-Q_-=1$, $\bA=\bE_\perp=0$. We describe them by the
standard Schr\"odinger kinetic energy 
(\ref{eq:schkin}). If we have $N_+$ positively charged particles and 
 $N_-$ negatively charged particles we may write the Hamiltonian as
$$
  H_{N_+,N_-}=\sum_{j=1}^{N_++N_-}
  -\frac1{2}\Delta_j+\sum_{1\leq i<j\leq
    N_++N_-}\frac{e_ie_j}{|r_i-r_j|},
$$
where $e_j=1$ if $j=1,\ldots,N_+$ and 
 $e_j=-1$ if $j=N_++1,\ldots,N_++N_-$. 
The Hilbert space is 
$$
\cH_{N_+,N_-}=P^+_{N_+}\bigotimes^{N_+}L^2(\R^3)\otimes
P^+_{N_-}\bigotimes^{N_-} L^2(\R^3).$$
This system is not stable of the second kind, in fact, the energy
behaves like the number of particles to the $7/5$-th power. The
following precise asymptotics was conjectured by Dyson in \cite{dy}.
\begin{theorem}[Dyson's formula]\label{thm:dyson}\hfill\\
Let
$$
E(N)=\inf_{N_++N_-=N}\inf\{\,\langle \psi,H_{N_+,N_-}\psi\rangle\ |
\ \psi\in\cH_{N_+,N_-}\cap C_0^\infty(\R^{3(N_++N_-)}),\ \|\psi\|=1\}
$$
then as $N\to\infty$
\begin{equation}\label{eq:dyson}
  \lim_{N\to\infty}\frac{E(N)}{N^{7/5}}=\inf\biggl\{\mfr{1}{2}\int_{\R^3}|\nabla\Phi|^2-I_0\int_{\R^3} \Phi^{5/2}\ 
  \biggr|\ 0\leq \Phi,\ \int_{\R^3}\Phi^2=1\biggr\},
\end{equation}
with $I_0$ given by 
\begin{equation}\label{eq:I0def}
  I_0=(2/\pi)^{3/4}\int_0^\infty1+x^4-x^2\left(x^4+2\right)^{1/2}dx=
  \frac{4^{5/4}\Gamma(3/4)}{5\pi^{1/4}\Gamma(5/4)} .
\end{equation}
\end{theorem}
{F}rom the Sobolev inequality (Theorem~\ref{thm:diasob})
it follows that the inf on the right of (\ref{eq:dyson}) is finite.
In \cite{dy} Dyson proved an upper bound on $E(N)$ of the form 
$-cN^{7/5}$ and thus indeed proved the instability of the second
kind. In \cite{CLY} a lower bound of the form $-CN^{7/5}$ was
established thus concluding that $7/5$ is the correct power.  
The theorem was finally proved in \cite{ls,so}.
In \cite{l79} Lieb proved that if the positively charged particles
have infinite mass then the energy is much smaller, indeed, 
bounded above by $-CN^{5/3}$ a corresponding lower bound had already
been proved in \cite{dl1,dl2}.

The proof of Theorem~\ref{thm:dyson} relies on an application of
Bogolubov's theory of superfluidity \cite{BO}. The charged system, in
fact, forms a superfluid state. 

Dyson's formula (\ref{eq:dyson}) is proved by establishing the
corresponding two inequalities. Establishing the lower bound is
technically very involved and is beyond the scope of these notes. It
is the content of the paper \cite{ls}. We will here give a brief
sketch of the proof of the upper bound from \cite{so}.
The upper bound is proved by finding an appropriate trial state. 
Here we are guided by Bogolubov's theory. 

It turns out that it is significantly easier to write down a 
grand canonical trial state than a canonical state. 
We are, however, interested in a canonical state. This will not be
a serious problem as we will eventually be able to show that the state
we construct is sharply peaked around the average 
particle number. We will ignore this point here and simply work with the
grand canonical state. We refer the reader to \cite{so} for 
details.

Another simplification is to consider the two species of bosons as one
species with two internal degrees of freedom corresponding to the two
signs of the charge. Constructing a trial state in this space will
correspond to averaging over states with different numbers of
positively and negatively charged particles.

We are thus considering the Fock space 
$\cF^+=\cF^+(L^2(\R^3)^2)$. We write a function $f\in L^2(\R^3)^2$, as
$f=f(r,e)$, where $e=\pm1$ is the sign of the charge. 
Let $|{\bf 0}\rangle$ be the vacuum vector in $\cF^+$.

In constructing a bosonic trial state the first guess is to put all
particles in the same one-particle state, i.e., 
to have a condensate. Let this state be represented by the
(normalized) vector
$\xi\in L^2(\R^3)^2$. Introduce first the normalized grand canonical 
vector
$$
|\Xi\rangle=\exp\left(-\frac{N}{2}+\sqrt{N}a_+^*(\xi)\right)|{\bf
  0}\rangle
=\sum_{n=0}^\infty e^{-N/2}\frac{N^{n/2}}{n!}a_+^*(\xi)^n|{\bf
  0}\rangle.
$$
The corresponding state is an average over states with varying
occupation in the condensate $\xi$. The average particle number in
$\xi$ is $ \langle\Xi|a_+^*(\xi)a_+(\xi)|\Xi\rangle=N$ and the
variance is also
$$
\langle\Xi|
(a_+^*(\xi)a_+(\xi))^2|\Xi\rangle-
\langle\Xi|a_+^*(\xi)a_+(\xi)|\Xi\rangle^2
=N.
$$
Thus this state is peaked around particle number $N$ with a standard
deviation $\sqrt{N}$. 

There is a unitary operator $U$ on $\cF^+$ such that 
$$
U^*a^*_+(f)U=a^*_+(f)+\sqrt{N}\langle\xi,f_\alpha\rangle.
$$
Using this unitary we may also write $|\Xi\rangle=U|{\bf
  0}\rangle$.

A pure condensate like this will however not give
the correct state. It is important to build pair excitations too. 
This is achieved as follows. Let $\{f_\alpha\}_{\alpha=0}^\infty$ be an
orthonormal family in $L^2(\R^3)^2$ (they will represent the pair
states). The normalized vector $\Psi\in\cF^+$ representing our trial
state may be abstractly
written
\begin{eqnarray}\label{eq:Bogstate}
  \Psi&=&\prod_{\alpha=0}^\infty(1-\lambda_\alpha^2)^{1/4}
  \exp\left(\sum_{\alpha=0}^\infty-\frac{\lambda_\alpha}{2}\left(a_+^*(f_\alpha)
      -\sqrt{N}\langle\xi,f_\alpha\rangle\right)^2\right)|\Xi\rangle\\
  &=&U\prod_{\alpha=0}^\infty(1-\lambda_\alpha^2)^{1/4}
  \exp\left(\sum_{\alpha=0}^\infty-\frac{\lambda_\alpha}{2}a_+^*(f_\alpha)^2
     \right)|{\bf
  0}\rangle.
  \nonumber
\end{eqnarray} 
We have introduced parameters 
$0<\lambda_\alpha<1$ with $\sum_{\alpha=0}^\infty\lambda_\alpha^2<\infty$
to control the occupations in the pair states. 
For simplicity we will assume that $\xi$ and
$\{f_\alpha\}_{\alpha=0}^\infty$ are real functions. 

We encode the information about the pair states in the positive
semi-definite trace class operator on $L^2(\R^3)^2$
\begin{equation}
\gamma=\sum_{\alpha=0}^\infty\frac{\lambda_\alpha^2}{1-\lambda_\alpha^2}
|f_\alpha\rangle\langle f_\alpha|.
\end{equation}
In terms of this operator a lengthy but straightforward calculation
shows that 
$$
\left\langle\Psi, (a_+^*(f)-\sqrt{N}\langle\xi,f\rangle)
 (a_+(g)-\sqrt{N}\langle g,\xi\rangle)\Psi\right\rangle
=\langle g,\gamma f\rangle,
$$
(the inner product on the left is in $\cF_+$ and the one on the right
is in $L^2(\R^3)^2$)
and 
$\left\langle\Psi, (a_+^*(f)-\sqrt{N}\langle\xi,f\rangle)\Psi\right\rangle=0$.
In particular,
$$
\left\langle\Psi, a_+^*(f)a_+(g)\Psi\right\rangle=\langle
g,(N|\xi\rangle\langle\xi|
+\gamma) f\rangle.
$$
Or equivalently using the field operators from
Section~\ref{sec:2ndquant}
\begin{equation}\label{eq:Psi1pdm}
  \left\langle\Psi, \phi(r,e)^*\phi(r',e')\Psi\right\rangle=N\xi(r,e)\xi(r',e')
  +\gamma(r,e;r',e')
\end{equation}
where $\gamma(r,e;r',e')$ is the integral kernel of $\gamma$. 
Likewise,
$$
\left\langle\Psi, a_+^*(f)a_+^*(g)\Psi\right\rangle=\left\langle
g,\left(N|\xi\rangle\langle\xi|
-\sqrt{\gamma(\gamma+1)}\right) f\right\rangle,
$$
or
\begin{equation}\label{eq:Psi1pdmoff}
  \left\langle\Psi, \phi(r,e)^*\phi(r',e')^*\Psi\right\rangle=
  N\xi(r,e)\xi(r',e')
  -\sqrt{\gamma(\gamma+1)}(r,e;r',e').
\end{equation}
Moreover, the state represented by $\Psi$ satisfies Wick's formula,
which for the 4-point function reads
\begin{eqnarray*}
  \lefteqn{\left\langle\Psi,\prod_{j=1}^4(a_+^\#(g_j)-\sqrt{N}\langle
      g_j,\xi\rangle^\#)\Psi\right\rangle}&&\\&=& 
  \left\langle\Psi,\prod_{j=1,2}(a_+^\#(g_j)- \sqrt{N}\langle g_j,\xi\rangle^\#)\Psi\right\rangle
  \left\langle\Psi,\prod_{j=3,4}(a_+^\#(g_j)- \sqrt{N}\langle g_j,\xi\rangle^\#)
    \Psi\right\rangle\\
  &&+\left\langle\Psi,\prod_{j=1,3}(a_+^\#(g_j)- \sqrt{N}\langle g_j,\xi\rangle^\#)\Psi\right\rangle
  \left\langle\Psi,\prod_{j=2,4}(a_+^\#(g_j)- \sqrt{N}\langle g_j,\xi\rangle^\#)
    \Psi\right\rangle\\
  &&+\left\langle\Psi,\prod_{j=1,4}(a_+^\#(g_j)- \sqrt{N}\langle g_j,\xi\rangle^\#)\Psi\right\rangle
  \left\langle\Psi,\prod_{j=2,3}(a_+^\#(g_j)- \sqrt{N}\langle g_j,\xi\rangle^\#)
    \Psi\right\rangle.
\end{eqnarray*}
Here $\#$ refers to either a $*$ (interpreted as complex conjugation
on scalars) or no $*$.
In particular, since $\xi$ is real this gives 
\begin{eqnarray}
  \lefteqn{\Big\langle\Psi,\big(\phi(r,e)^*-\sqrt{N}\xi(r,e)\big)\big(\phi(r',e')^*-\sqrt{N}\xi(r',e')\big)
    }\hspace{3cm}&&\nonumber\\&\times&
    \big(\phi(r',e')-\sqrt{N}\xi(r',e')\big)\big(\phi(r,e)-\sqrt{N}\xi(r,e)\big)
    \Psi\Big\rangle\nonumber\\
    &=&|\sqrt{\gamma(\gamma+1)}(r,e;r',e')|^2+|\gamma(r,e;r',e')|^2\nonumber
    \\&&
    +\gamma(r,e;r,e)\gamma(r',e';r',e')\label{eq:Psi4point}
\end{eqnarray}
Armed with these identities we can calculate the expectation of the
energy in the state represented by $\Psi$. 

First we will explain, for the special case of the charged Bose
system, how to choose the condensate function $\xi$ and
the trace class operator $\gamma$. More precisely, we will specify
their charge dependence. We set
\begin{equation}\label{eq:xi0}
  \xi(r,e)=\sqrt{\frac12}\xi_0(r),
\end{equation}
where $\xi_0$ is a real normalized function in $L^2(\R^3)$.
Thus the condensate function does not depend on the charge. 
The operator $\gamma$ on $L^2(\R^3)^2=L^2(\R^3)\otimes\C^2$ will
be chosen to have the form
$$
\gamma=\gamma_0\otimes\frac12\left(\begin{array}{cc}
1&-1\\-1&1
  \end{array}\right),
  $$
where $\gamma_0$ is a positive trace-class operator on $L^2(\R^3)$.
Put differently, the integral kernel of $\gamma$ is chosen to be
\begin{equation}\label{eq:gamma0}
\gamma(r,e;r',e')=\frac12 ee'\gamma_0(r;r').
\end{equation}
The charge part of this operator is a rank one operator and thus we
also have
$$
\sqrt{\gamma(\gamma+1)}=\sqrt{\gamma_0(\gamma_0+1)}\otimes\frac12\left(\begin{array}{cc}
1&-1\\-1&1
  \end{array}\right).
$$
It is now straightforward to calculate the expectation of the Coulomb
potential in the state
represented by $\Psi$. {F}rom (\ref{eq:Psi1pdm}--\ref{eq:gamma0}) we obtain
\begin{eqnarray*}
\lefteqn{\left\langle\Psi,\bigoplus_{M=0}^\infty
\sum_{1\leq i<j\leq M}\frac{e_ie_j}{|r_i-r_j|}\Psi\right\rangle=}&&\\
&&
\left\langle\Psi,\frac12\sum_{ee'=\pm1}\iint ee'\phi(r,e)^*\phi(r',e')^*|r-r'|^{-1}
\phi(r,e)\phi(r',e')\Psi\right\rangle
\\&=&N\Tr_{L^2(\R^3)}\left(\cK(\gamma_0-\sqrt{\gamma_0(\gamma_0+1)})\right).
\end{eqnarray*} 
Here $\cK$ is the operator with integral kernel
$$
\cK(r,r')=\xi_0(r)|r-r'|^{-1}\xi_0(r').
$$
The total energy expectation is 
\begin{eqnarray*}
  \Big\langle\Psi,\bigoplus_{N_+,N_-=0}^\infty
  H_{N_+,N_-}\Psi\Big\rangle&=&\frac{N}2\int|\nabla\xi_0|^2\\&&
  +\frac12\Tr(-\Delta\gamma_0)
  +N\Tr\Big(\cK\big(\gamma_0-\sqrt{\gamma_0(\gamma_0+1)}\big)\Big).
\end{eqnarray*}
The final step in the argument is to minimize the above expression
over $\gamma_0$. More precisely this is done in a semiclassical
approximation. We will only sketch this argument. The rigorous
argument can again be found in \cite{so}. 
 We assume that $\gamma_0$ is the quantization of a classical symbol 
$f(r,p)\geq0$. The semiclassical approximation to the energy is then
\begin{eqnarray*}
\lefteqn{\frac{N}2\int|\nabla\xi_0|^2}&&\hspace{3cm}\\&&
+(2\pi)^{-3}\iint \frac{p^2}{2}f(r,p)+4\pi N|p|^{-2}\xi_0(r)^2\big(f(r,p)-\sqrt{f(r,p)(f(r,p)+1)}\big)drdp.
\end{eqnarray*}
Minimizing this expression over $f(r,p)$ and performing the $p$
integration gives
$$
\frac{N}2\int|\nabla\xi_0|^2-I_0N^{5/4}\int\xi_0(r)^{5/2}dr,
$$
where $I_0$ is given in (\ref{eq:I0def}).
If we introduce the rescaling $\xi_0(r)=N^{3/10}\Phi(N^{1/5}r)$, where
$\Phi$ is also normalized then the energy expression above becomes
$$
N^{7/5}\Big(\int|\nabla\Phi|^2-I_0\int\Phi^{5/2}\Big),
$$
which is exactly the expression conjectured by Dyson for the energy.

Note that the instability is also reflected in the shrinking of the
linear dimension of the state with increasing $N$. According to the
scaling of $\xi_0$ above, the linear dimension of the state behaves
like $\sim N^{-1/5}$.


\begin{thebibliography}{00}

\bibitem{b} Baxter, John R.\ Inequalities for potentials of particle
  systems. {\it Illinois J.\ Math.}, {\bf 24}, (1980).


\bibitem{BO} Bogolubov, N.,  On the theory of superfluidity,
{\it J.\  Phys. (U.S.S.R.)}  {\bf 11}, 23,
(1947).

\bibitem{bfg} Bugliaro, L.\ and Fr\"ohlich, J.\ and , Graf, G.M.,
Stability of quantum electrodynamics with nonrelativistic matter,
{\it Phys.\ Rev.\ Lett.}, {\bf 77}no.\  17, 3494--3497, (1996).

\bibitem{C} Conlon, Joseph G.,
The ground state energy of a classical gas.
{\it Comm.\ Math.\ Phys.} {\bf 94}, no. 4, 439–-458 (1984).

\bibitem{CLY} Conlon, Joseph G.\ and Lieb, Elliott H.\ and 
Yau, Horng-Tzer,
The $N^{7/5}$ law for charged bosons.
{\it Comm.\ Math.\ Phys.}, {\bf 116} , no. 3, 417–-448 (1988).

\bibitem{dy} Dyson, Freeman J., Ground state energy of a finite system
  of charged particles, {\it Jour.\  Math.\  Phys.}, {\bf 8}, 1538--1545
  (1967).

\bibitem{dl1} Dyson, Freeman J. and Lenard, Andrew, Stability of
  matter. I, {\it Jour.\  Math.\  Phys.}, {\bf 8}, 423--434,
  (1967).

\bibitem{dl2} Dyson, Freeman J. and Lenard, Andrew, Stability of
  matter. II, {\it Jour.\  Math.\  Phys.}, {\bf 9}, 698--711, (1968).

\bibitem{Fe} Fefferman, Charles, The Thermodynamic Limit for a
  Crystal. 
{\it Comm.\ Math.\ Phys.}, {\bf 98}, 289–-311 , (1985).
 
\bibitem{fef} Fefferman, Charles, Stability of matter with magnetic
  fields, {\it CRM Proc.\  Lecture Notes} {\bf 12}, 119--133 (1997).

\bibitem{ffg} Fefferman, Charles, Fr{\"o}hlich, J{\"u}rg, and Graf,
  Gian Michele, Stability of nonrelativistic quantum mechanical matter
  coupled to the (ultraviolet cutoff) radiation field, {\it Proc. Natl.
    Acad. Sci. USA} {\bf 93}, 15009--15011 (1996); Stability of
  ultraviolet cutoff quantum electrodynamics with non-relativistic
  matter, {\it Comm.\ Math.\ Phys.}, {\bf 190}, 309--330, (1997).

\bibitem{fdl} {Fefferman, Charles and de la Llave, Rafael}, {Relativistic
    stability of matter.\  {I}}, {\it Rev.\  Mat.\  Iberoamericana}, {\bf
    2}, {119--213}, (1986).
  
\bibitem{ES} Erd\H{o}s, L\'aszl\'o and Solovej, Jan Philip, 
The kernel of Dirac operators on ${\bf S}^3$ and ${\bf R}^3$. 
{\it Rev.\ Math.\ Phys.} {\bf 13},
no. 10, 1247--1280, (2001).

\bibitem{GS} Graf,Gian Michele and Schenker, Daniel,
On the molecular limit of Coulomb gases.
Comm.\ Math.\ Phys., {\bf 174} , no. 1, 215--227, (1995).
 
\bibitem{HLS} Hainzl, Christian and Lewin, Mathieu and Solovej, Jan
  Philip, The mean-field approximation in quantum
    electrodynamics: the no-photon case, {\it Comm.\ Pure Appl.\ Math.\
  }, {\bf 60}, 546--596, (2007).

\bibitem{HLS2}  Hainzl, Christian and Lewin, Mathieu and Solovej, Jan
  Philip,{The Thermodynamic Limit of Quantum Coulomb
    Systems. Part I. General Theory.} {\it Advances in Mathematics}. {\bf
    221},  454--487, (2009). 

\bibitem{HLS3}
  Hainzl, Christian and Lewin, Mathieu and Solovej, Jan
  Philip, {The Thermodynamic Limit of Quantum Coulomb
    Systems. Part II. Applications.} {\it Advances in Mathematics}. {\bf
    221}, 488--546, (2009). 

\bibitem{l-rev} Lieb, Elliott H., The Stability of Matter, {\it Rev.\ 
    Mod.\ Phys.}, {\bf 48}, 553--569, (1976).


\bibitem{l79} Lieb, Elliott H., The $N^{5/3}$ Law for Bosons, {\it
    Phys.\ Lett.}, {\bf 70A}, 71--73, (1979).

\bibitem{LLbook} Lieb, Elliott H. and Loss, Michael, \emph{Analysis},
Graduate Studies in Mathematics, {\bf 14}, American Mathematical
Society (2001)

\bibitem{LSbook} Lieb, Elliott H. and Seiringer, Robert, \emph{The
    stability of matter in quantum mechanics}, Cambridge University
  Press, Cambridge (2010).

\bibitem{lile} Lieb, Elliott H.\ and Lebowitz, Joel L.\ , The
  constitution of matter: Existence of thermodynamics for systems
  composed of electrons and nuclei. {\it Advances in Mathematics}, {\bf 9},
  316--398, (1972).
 
\bibitem{lilo2}{Lieb, Elliott H.\ and Loss, Michael}, {Stability of a
    model of relativistic quantum electrodynamics}, {\it Comm.\ 
    Math.\ Phys.} {\bf 228}, {561--588} (2002).
  
\bibitem{lls-tf} Lieb, Elliott H., Loss, Michael, and Siedentop,
  Heinz, Stability of Relativistic Matter via Thomas-Fermi Theory,
  {\it Helv.~Phys.~Acta}, {\bf 69}, 974--984, (1996).

\bibitem{lls} Lieb, Elliott H.\ and Loss, Michael and Solovej, Jan
  Philip, Stability of matter in magnetic fields, {\it Phys.\ Rev.\ 
    Lett.} {\bf 75}, 985--989 (1995).
  
  
\bibitem{lss} Lieb, Elliott H. and Siedentop, Heinz and Solovej, Jan
  Philip, Stability and instability of relativistic electrons in
  classical electromagnetic fields, {\it Jour.\ Stat.\ Phys.} {\bf
    89}, 37--59 (1997).
   
 
\bibitem{ls} {Lieb, Elliott H.\ and Solovej, Jan Philip}, {G}round
  state energy of the two-component charged {B}ose gas,
  {\it Comm.\ Math.\ Phys.}, {\bf 252}, 485 -- 534, (2004).
  
\bibitem{lt} Lieb, Elliott H.\ and Thirring, Walter E.\ , Bound for
  the kinetic energy of fermions which proves the stability of matter,
  {\it Phys.\ Rev.\ Lett.} {\bf 35}, 687--689, (1975).

\bibitem{ly} Lieb, Elliott H.\ and Yau, Horng-Tzer, The stability and
  instability of relativistic matter, {\it Comm.\ Math.\ Phys.}, {\bf
    118}, 177--213, (1988).
  
\bibitem{loya} Loss, Michael and Yau, Horng-Tzer, {Stability of
    {C}oulomb systems with magnetic fields.\ {III}.\ {Z}ero energy
    bound states of the {P}auli operator}, {\it Comm.\ Math.\ Phys.}
  {\bf 104}, {283--290} (1986).

\bibitem{on} Onsager, Lars, Electrostatic Interaction of Molecules,
  {\it Jour.\ Phys.\ Chem.} {\bf 43}, 189--196, (1939).


\bibitem{spohn}Spohn, Herbert, \emph{Dynamics of charged particles and
    their radiation field}, Cambridge University Press, Cambridge
  (2004).

\bibitem{so} Solovej, Jan Philip, Upper Bounds to the Ground State
  Energies of the One- and Two-Component Charged Bose Gases.
{\it Comm.\ Math.\ Phys.}, {\bf 266}, No 3, 797--818, (2006).


\end{thebibliography}
\end{document}